\documentclass[letterpaper, 10 pt, conference]{ieeeconf}  

\IEEEoverridecommandlockouts                              

\usepackage{amssymb}
\usepackage{graphicx}
\usepackage{epstopdf}
\usepackage{amsmath}
\usepackage{mathtools, cuted}
\usepackage{multirow}
\usepackage{array}
\usepackage{bm}
\usepackage{url}

\usepackage{tabularx}
\usepackage{stackengine}
\setstackEOL{\\}

\newtheorem{problem}{\textbf{Problem}}
\newtheorem{definition}{\textbf{Definition}}

\newtheorem{lemma}{\textbf{Lemma}}
\newtheorem{remark}{\textbf{Remark}}

\newtheorem{proposition}{\textbf{Proposition}}

\title{\LARGE \bf
Model-based Randomness Monitor for Stealthy Sensor Attacks
}
\author{Paul J Bonczek, Shijie Gao, and Nicola Bezzo
\thanks{Paul J Bonczek, Shijie Gao, and Nicola Bezzo are with the Charles L. Brown Department of Electrical and Computer Engineering, and Link Lab, University of Virginia, Charlottesville, VA 22904, USA. Email: {\tt \{pjb4xn, sg9dn, nb6be\}@virginia.edu}}
}

\newcommand*{\N}{\mathbb{N}}
\newcommand*{\R}{\mathbb{R}}

\renewcommand{\baselinestretch}{0.94}
\begin{document}

\maketitle
\thispagestyle{empty}
\pagestyle{empty}

\begin{abstract}

Malicious attacks on modern autonomous cyber-physical systems (CPSs) can leverage information about the system dynamics and noise characteristics to hide while hijacking the system toward undesired states. Given attacks attempting to hide within the system noise profile to remain undetected, an attacker with the intent to hijack a system will alter sensor measurements, contradicting with what is expected by the system's model. To deal with this problem, in this paper we present a framework to detect non-randomness in sensor measurements on CPSs under the effect of sensor attacks. 
Specifically, we propose a run-time monitor that leverages two statistical tests, the {\em Wilcoxon Signed-Rank} test and {\em Serial Independence Runs} test to detect inconsistent patterns in the measurement data. For the proposed statistical tests we provide formal guarantees and bounds for attack detection. We validate our approach through simulations and experiments on an unmanned ground vehicle (UGV) under stealthy attacks and compare our framework with other anomaly detectors.

\end{abstract}

\section{Introduction} \label{sec:introduction}
Modern autonomous systems are fitted with multiple sensors, computers, and networking devices that make them capable of many applications with little/no human supervision. Autonomous navigation, transportation, surveillance, and task oriented jobs are becoming more common and ready for deployment in real world applications especially in the automotive, industrial, and military domains. These enhancements in autonomy are possible thanks to the tight interaction between computation, sensing, communications, and actuation that characterize cyber-physical systems (CPSs). These systems are however vulnerable and susceptible to cyber-attacks like sensor spoofing which can compromise their integrity and the safety of the surroundings. In the context of autonomous vehicle technologies, one of the most typical threats is {\em hijacking} in which an adversary is capable to administer malicious attacks with the intent of leading the system to an undesired state. An example of this problem was demonstrated by authors in \cite{YachtSpoof} in which GPS data were spoofed to slowly drive a yacht off the intended route. 

If we look at the specific architecture of these robotic systems, typical autonomous applications employ go-to-goal and trajectory tracking and if one or more on-board sensors are compromised, system behavior can become unreliable. These vehicles typically have well studied dynamics and their sensors have specific expected behaviors according to their characterized noise models. 
An attacker that wants to perform a malicious hijacking can create non-random patterns or add small biases in the measurements to slowly push the system towards undesired states, for example creating undesired deviations as depicted in Fig.~\ref{fig:IntroImage}, all while remaining hidden within the system's and sensors noise profile. Hence, in order for an attacker to hijack the system with stealthy attack signals, a violation to the expected random behavior of the sensor measurements must occur.

\begin{figure}[t!]
\centering
\includegraphics[width=0.47\textwidth]{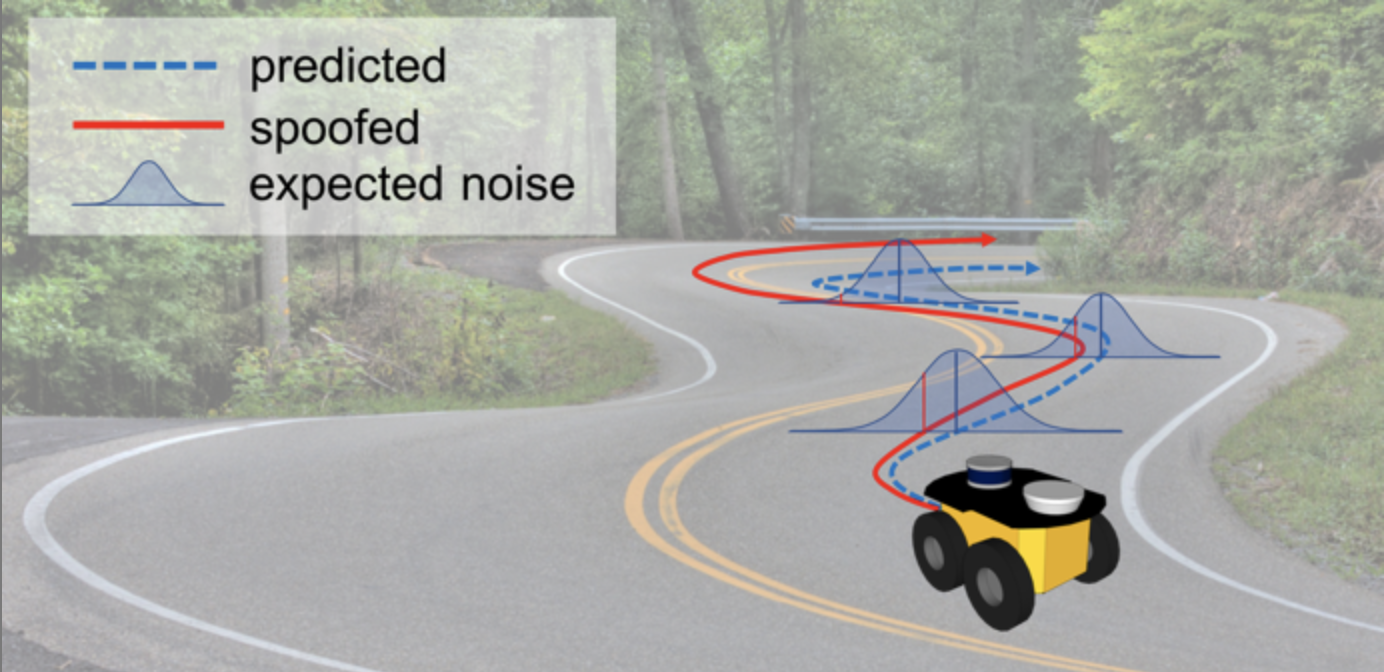}
\vspace{-5pt}
\caption{A pictorial representation of the problem discussed in this paper in which a cyber-attack is able to hijack a vehicle into unsafe states while remaining hidden within the noise profile of its sensors.}
\vspace{-16pt}
\label{fig:IntroImage}
\end{figure}

With these considerations and problem in mind, in this work, we leverage the known characteristics of the \textit{residual} -- the difference between sensor measurements and state prediction -- to build a run-time monitor to detect non-random behaviors. To monitor randomness, the non-parametric statistical Wilcoxon Signed-Rank \cite{Wilcoxon1} and Serial Independence Runs \cite{serial_test} tests are applied to individual sensors to determine if their measurements are being received randomly. The Wilcoxon test is an indicator of whether the residual is symmetric over its expected value, whereas the Serial Independence runs test indicates whether the sequences of residuals are arriving in a random manner.
Thus, the main objective of this work is to find hidden attacks exhibiting non-random behavior within the noise. Given the nature of the non-parametric statistical tests that we propose, only random behavior of the residual is considered here, leaving the magnitude bounds of the residual un-monitored. Several detectors providing magnitude bounds on attacks have been already researched in the literature, thus in this work we also present a framework to combine existing approaches for magnitude bound detection with the proposed randomness monitor. In doing so, our approach improves the state-of-the-art attack detection by adding an extra layer of checks.

\vspace{-3pt}

\subsection{Related Work}
\label{sec:Related Work}

This work builds on previous research considering deceptive cyber-attacks to {\em hijack} a system by injecting false data to sensor measurements while trying to remain undetected \cite{BadData}. Many of the previous works use the residual for detection, which gives clues whether sensor measurements are healthy (uncompromised). Previous works characterizing the effects of stealthy sensor attacks on the Kalman filter can be found in \cite{KF_attack,bezzo_SE}. Similarly, authors in \cite{BadData,CST1} discuss how stealthy, undetectable attacks can compromise closed-loop systems, causing state and system dynamic degradation

Several procedures and techniques that analyze the residual for attack detection exist, one of which is the Sequential Probability Ratio Testing (SPRT) \cite{SPRT2} that tests the sequence of incoming residuals one at a time by taking the log-likelihood function (LLF). 
The Cumulative Sum (CUSUM) procedure proposed in \cite{CUSUM1} and \cite{CUSUM_Journal} leverages the known characteristics of the residual covariance and sequentially sums the residual error to find changes in mean of the distribution. Compound Scalar Testing (CST) in \cite{CST1} is another technique which is computationally friendly by reducing the residual vector with the known residual covariance matrix into a scalar value with $\chi^2$ distribution. An improvement of CST in \cite{CST2} is made by including a coding matrix to sensor outputs that is unknown to attackers, then an iterative optimization algorithm is used to solve for a transform matrix to detect stealthy attacks. Similar to our work where monitors are placed on individual sensors, the authors in \cite{Trust_model} propose a Trust-based framework for sensor sets by ``side-channel" monitors to provide a weight for trustworthiness to determine whether sensors have been compromised. Other works have proposed attack resiliency by leveraging information from redundant sensing. In \cite{pajic_SE}, authors solve to reconstruct the state estimate of stochastic systems using an $l_0$ optimization problem when less than half of the sensors are compromised. Different from these previous works, we build a framework to monitor sensor measurements to find previously undetectable attacks by searching for non-random behavior.

The remainder of this work is organized as follows. In Section \ref{sec:preliminaries} we begin with system, estimator models and problem formulation, followed by the description of our Random Monitor framework in Section \ref{sec:framework}. In Section \ref{sec:State Degradation} an analysis of worst-case stealthy attacks and characterization of the effects on system performance is presented. Finally, in Section \ref{sec:Results} we demonstrate through simulations and experiments the performance of our framework augmented with boundary detectors before drawing conclusions in Section~\ref{sec:conclusion}.
\section{Preliminaries \& Problem Formulation} \label{sec:preliminaries}

In this work we consider autonomous systems whose dynamics can be described by a discrete-time linear time-invariant (LTI) system in the following form:
\vspace{-2pt}
\begin{equation}\label{eq:system1}
\begin{split}
\bm{x}_{k+1} &= \bm{A} \bm{x}_k + \bm{B} \bm{u}_k + \bm{\nu}_k \\
	\bm{y}_k&=\bm{C} \bm{x}_k + \bm{\eta}_k ,
\end{split}
\vspace{-4pt}
\end{equation}
with $\bm{A} \in \R^{n\times n}$ the state matrix, $\bm{B} \in \R^{n\times m}$ the input matrix, and $\bm{C} \in \R^{s\times n}$ the output matrix with the state vector $\bm{x}_k \in \R^{n}$, system input $\bm{u}_k \in \R^{m}$, output vector $\bm{y}_k \in \R^{s}$ providing measurements from $s$ sensors from the set $\mathcal{\bm{S}}=\{1,2,\dots,s\}$, and sampling time-instants $k \in \N$. Process and measurement noises are i.i.d. multivariate zero-mean Gaussian uncertainties $\bm{\nu} = \mathcal{N}(0,\bm{Q}) \in \R^n$ and $\bm{\eta} = \mathcal{N}(0,\bm{R}) \in \R^s$ with covariance matrices $\bm{Q} \in \R^{n\times n}, \bm{Q} \geq 0$ and $\bm{R} \in \R^{s\times s}, \bm{R} \geq 0$ and are assumed static.

During operations, a standard Kalman Filter (KF) is implemented to provide a state estimate $\hat{\bm{x}}_k \in \R^n$ in the form
\vspace{-3pt}
\begin{equation}\label{eq:Kalman}
	\hat{\bm{x}}_{k+1} = \bm{A} \hat{\bm{x}}_k + \bm{B} \bm{u}_k + \bm{L}(\bm{y}_k - \bm{C}\hat{\bm{x}}_k), 
\vspace{-2pt}
\end{equation}
where the Kalman gain matrix $\bm{L} \in \R^{n \times s}$ is 
\vspace{-2pt}
\begin{equation}
\label{eq:SteadyState_P_K}
	\bm{L} = \bm{A}\bm{P}\bm{C}^T(\bm{R} + \bm{C}\bm{P}\bm{C}^T)^{-1},
\vspace{-3pt}
\end{equation}
therefore, we assume that the KF is at steady state, i.e., $\lim_{k\to \infty} \bm{P}_k~=~\bm{P}$. The estimation error of the KF is defined as $\bm{e}_k = \bm{x}_k - \hat{\bm{x}}_k$ while its {\em residual} $\bm{r}_k$ is given by
\vspace{-2pt}
\begin{equation}
\label{eq:Residual}
	\bm{r}_k = \bm{y}_k - \bm{C}\hat{\bm{x}}_k = \bm{C}\bm{e}_k + \bm{\eta}_k,
	\vspace{-2pt}
\end{equation}
The covariance of the residual \eqref{eq:Residual} is defined as
\vspace{-2pt}
\begin{equation}
\label{eq:Residual_Covariance}
	\bm{\Sigma} = \mathrm{E}[\bm{r}_{k+1}\bm{r}_{k+1}^T]  = \bm{R} + \bm{C}\bm{P}\bm{C}^T \in \R^{s \times s}.
	\vspace{-2pt}
\end{equation}

In the absence of sensor attacks, the residual for the $i^{th}$ sensor $r_{k,i}, i \in \mathcal{S}$ follows a Gaussian distribution $r_{k,i} \sim \mathcal{N}(0,\sigma_i^2)$ where $\sigma_i^2$ is the $i^{th}$ diagonal element of the residual covariance matrix $\bm{\Sigma} \in \R^{s \times s}$ in \eqref{eq:Residual_Covariance} such that
\vspace{-3pt}
\begin{equation}
\label{eq:normal}
    \mathrm{E}[r_{k,i}] = 0, \text{ } \mathrm{Var}[r_{k,i}] = \sigma^2_i.
    \vspace{-2pt}
\end{equation}
We describe the system output considering sensor attacks as
\vspace{-4pt}
\begin{equation}\label{eq:output_equation}
	\bm{y}_k = \bm{C} \bm{x}_k + \bm{\eta}_k + \bm{\xi}_k ,
\vspace{-2pt}
\end{equation}
where $\bm{\xi}_k \in \R^s$ represents the sensor attack vector. Our proposed framework consists in adding a monitor on each sensor searching for non-random behavior of the sensor measurement residual, hence any sensor may be compromised.

\begin{definition}
A sensor measurement is random if:
\vspace{-2pt}
\label{consist_definition}
\begin{itemize}
\item{a sequence of residuals over a time window occurs in an unpredictable, pattern-free manner.}
\item{residuals have proper distributions over $\mathrm{E}[\bm{r}_k]$}.
\end{itemize}
\end{definition}

Since we are considering sensor spoofing, an attack signal $\bm{\xi}_k$ containing malicious data can disrupt randomness, causing measurements to display non-random behavior. Formally, the problem that we are interested in solving is:
\begin{problem} 
\label{problem1} {\textbf{Randomness of Measurements:}} 
Given the residual $\bm{r}_k$ between a measurement $\bm{y}_k$ and the corresponding prediction $\bm{C}\hat{\bm{x}}_k$ as defined in \eqref{eq:Residual}, find a policy to determine at run-time whether a sensor measurement is random, i.e., if any condition in Definition \ref{consist_definition} does not hold.
\end{problem}
\begin{section}{Randomness Monitoring Framework}
\label{sec:framework}

The overall cyber-physical system architecture including our Randomness Monitor framework is summarized in Fig.~\ref{fig:AttackDiagram}. The Randomness Monitor, augmented to any boundary detector providing magnitude bounds, is placed in the system feedback to monitor the residual sequence. 
\vspace{-7pt}
\begin{figure}[th!b]
\centering
\includegraphics[width=0.45\textwidth]{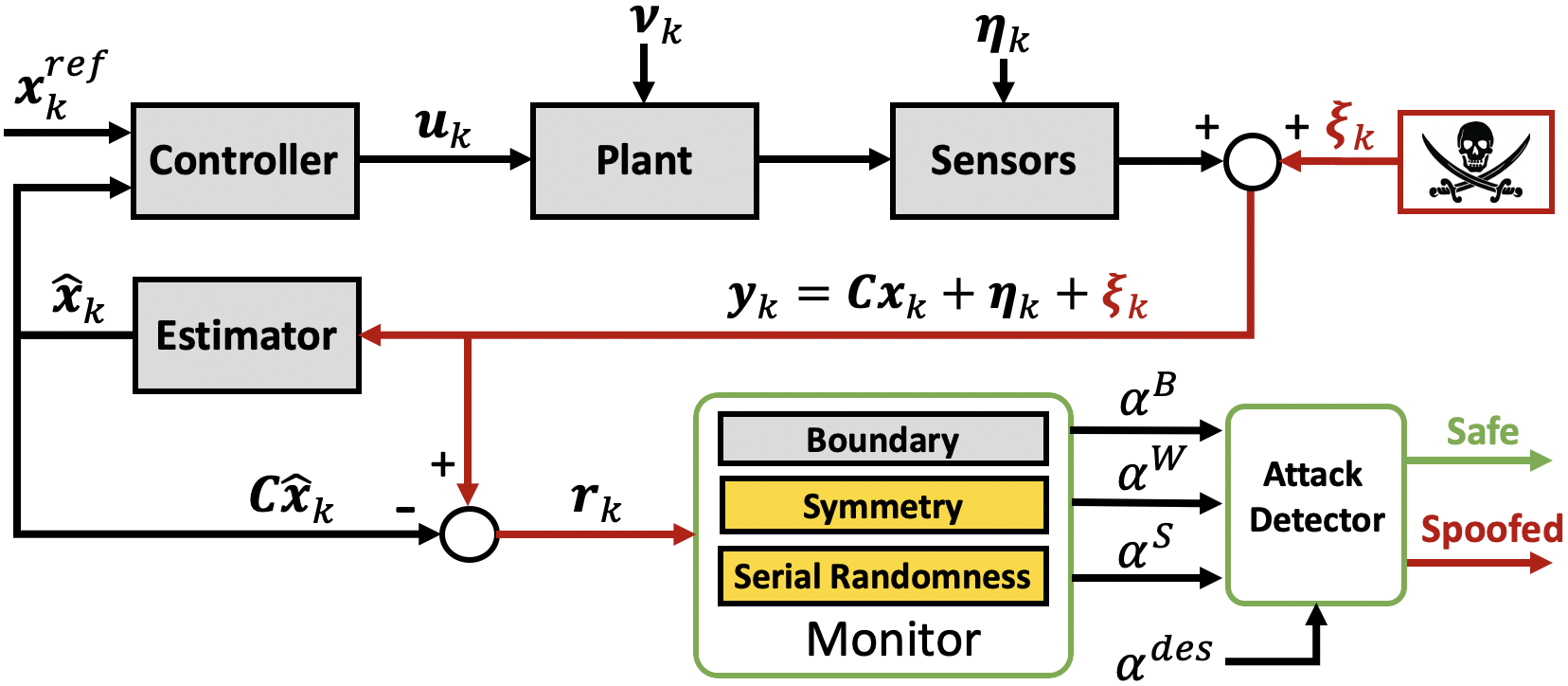}
\vspace{-4pt}
\caption{The architecture of a CPS while experiencing sensor attacks augmented with our monitoring technique.}
\label{fig:AttackDiagram}
\vspace{-6pt}
\end{figure}
We introduce a framework to monitor randomness of the residual sequence through two tests and provide tuning bounds for each to result in desired false alarm rates. 
From \eqref{eq:Residual}, the residual should have a symmetric distribution centered at zero and the sequence of residuals should arrive in a random order, having an absence of structure or patterns. For example, a continuously alternating pattern of ``negative" and ``positive" values, or a pattern of only ``negative" values would clearly not satisfy random sequences. 

Both tests operate online providing an alarm when the residual does not satisfy the conditions of each test. A desired false alarm rate $\alpha_i^{\text{des}} \in (0, 1)$ for each $i^{th}$ sensor is tuned for each test, and in the absence of sensor attacks, the observed alarm rate $\alpha_i \in [0, 1]$ for each test should match closely with the tuned desired value $\alpha_i \sim \alpha_i^{\text{des}}$.

\subsection{Residual Symmetry Monitor}
\label{sec:Wilcoxon}

To monitor whether the sequence of residuals are symmetrically distributed and zero-mean, we leverage the Wilcoxon Signed-Rank (WSR) test \cite{Wilcoxon1} as follows. A hypothesis test is formed by $\mathcal{H}_0$ for no attacks and $\mathcal{H}_a$ with attacks:
\vspace{-1pt}
\begin{equation}
\begin{split}
\label{eq:Wilcoxon_Hypothesis}
    \bigg\{ \begin{array}{l}
    \begin{aligned}
	&\mathcal{H}_0: \mathrm{E}[\bm{r}_k]=0 \textbf{ and } \bm{r}_k \text{ is symmetric}, \\
    &\mathcal{H}_a: \mathrm{E}[\bm{r}_k] \ne 0 \hspace*{6pt} \textbf{ or } \bm{r}_k \text{ is not symmetric}.
    \end{aligned}
    \end{array}
\end{split}
\vspace{-2pt}
\end{equation}
A monitor is built to check if the residual $\bm{r}_{k}$ sequence over a sliding monitoring window $T = (k - \ell + 1, k)$ for $\ell$ previous steps is symmetric. We denote the vector of residual sequences over the sliding window $T$ as $\bm{r}_{T}~=~(\bm{r}_{T,1},\dots,\bm{r}_{T,i},\dots,\bm{r}_{T,s})$ where the residual sequence for an $i^{th}$ sensor is $\bm{r}_{T,i} = (r_{k-\ell+1,i} ,\dots,r_{k,i})$. Following $\mathcal{H}_0$, we would expect that the number of positive and negative values of $\bm{r}_k$ over the monitoring window are equal. Additionally, a symmetric distribution indicates that the expected absolute magnitude of positive and negative residuals over a given window of length $\ell$ are equal,
\vspace{-2pt}
\begin{equation}
\label{eq:ExpectedMagnitude}
    \mathrm{E}[|\bm{r}_{T,i}^{+}|] = \mathrm{E}[|\bm{r}_{T,i}^{-}|], \text{ } i \in \mathcal{S},
    \vspace{-1pt}
\end{equation}
where $\mathrm{E}[|\bm{r}_{T,i}^{+}|]$ and $\mathrm{E}[|\bm{r}_{T,i}^{-}|]$ denote the expected absolute magnitude for positive and negative values of the residual $r_{k,i}$ within the window $T$ for any given $i^{th}$ sensor. In other words, we would expect the sum of absolute values from the residual to be equal for both the positive and negative values. The WSR test takes both the sign and magnitude of the residual into account to determine whether conditions satisfy $\mathcal{H}_0$. Large differences in the residual signs or signed magnitudes imply non-similar distributions, causing the test to reject the no attack assumption and triggering an alarm.

To perform the WSR test at each time step $k$, we first look at the $\ell$ number of residuals over the monitoring window $T$ of a given $i^{th}$ sensor, ranking the \textit{absolute values} of residuals $r_{T,i}$, starting with $rank =1$ for the smallest absolute value, $rank=2$ for the second smallest, and so on until reaching the largest absolute value with $rank=\ell$. Ranks of absolute values for positive (i.e. $|r_{T,i}^{+}|$) and negative (i.e.  $|r_{T,i}^{-}|$) residuals over the window $T$ are placed into the sets $\mathcal{R}_{k,i}^{+}$ and $\mathcal{R}_{k,i}^{-}$ at every time instance $k$, respectively.
\vspace{1pt}
\begin{remark}
\label{rem:Wilcoxon}
    For residuals equal to each other and not equal to $0$ (tied for the same rank), an average of the ranks that would have been assigned to these residuals is given to each of the tied values. Furthermore, residuals equal to $0$ are removed and $\ell$ is reduced accordingly. 
\end{remark}

Following, we compute the sum of ranks for both the positive and negative valued residuals,
\vspace{-2pt}
\begin{equation}
\label{eq:WilcoxonTestStat}
    W_{k,i}^{+} = \sum \mathcal{R}_{k,i}^{+}, \hspace*{10pt} W_{k,i}^{-} = \sum \mathcal{R}_{k,i}^{-}.
    \vspace{-2pt}
\end{equation}

Residuals with symmetric distributions have similar valued sum of ranks, i.e. $W_{k,i}^{+} \sim W_{k,i}^{-} $, whereas the sum of ranks in non-symmetric distributions are not similar $W_{k,i}^{+} \nsim W_{k,i}^{-}$ resulting in a rejection of $\mathcal{H}_0$ in \eqref{eq:Wilcoxon_Hypothesis}, which we will now discuss how to solve. Assuming a large window of size $\ell \geq 20$\footnote{For window length of smaller size, exact tables need to be used for probability distributions of the Wilcoxon Signed-Rank random variable \cite{Seigel}.} \cite{Seigel}, the Wilcoxon random variables $W_{k,i}^{+}$, $W_{k,i}^{-}$ converge to a Normal distribution (without attacks) as $\ell \to \infty$ and can be approximated to a standard normal distribution. The approximated expected value and variance of the two sum of ranks $W_{k,i}^{+}$ and $W_{k,i}^{-}$, denoted as $W_{k,i}^{\pm} = \{W_{k,i}^{+},W_{k,i}^{-}\}$ is
\begin{equation}
\label{eq:WilcoxonMeanVar}
    \begin{array}{ll}
    \vspace{1pt}
    \mathrm{E}[W_{k,i}^{\pm}] = \frac{\ell^2+\ell}{4}, \hspace{14pt} &
    \mathrm{Var}[W_{k,i}^{\pm}] = \frac{(\ell^2+\ell)(2\ell+1)}{24}.
    \end{array}
    \vspace{-2pt}
\end{equation}
The z-score of \eqref{eq:WilcoxonTestStat} for a given $i^{th}$ sensor is computed by
\vspace{-3pt}
\begin{equation}
\label{eq:WilcoxonZ}
    \hspace*{-6pt} Z_{k,i}^W \hspace*{-1pt} = \hspace*{-1pt} \frac{\min(W_{k,i}^{\pm}) \hspace*{-1pt} -\hspace*{-1pt} \mathrm{E}[W_{k,i}^{\pm}]}{\sqrt{\mathrm{Var}[W_{k,i}^{\pm}]}} \hspace*{-1pt}= \hspace*{-1pt} \frac{\min(W_{k,i}^{\pm}) \hspace*{-1pt} - \hspace*{-1pt} \frac{(\ell^2+\ell)}{4}}{\sqrt{ \frac{(\ell^2+\ell)(2\ell+1)}{24} } },
    \vspace{-3pt}
\end{equation}
and the p-value used to determine whether to reject the null hypothesis $\mathcal{H}_0$ (i.e. no attacks) is computed from \eqref{eq:WilcoxonZ} as
\vspace{-2pt}
\begin{equation}
\label{eq:P_approx_Wilcoxon}
	p_{k,i}^W = \Phi(|Z_{k,i}^W|) = 2\cdot\frac{1}{\sqrt{2\pi}}\int_{|Z_{k,i}^W|}^{\infty} \text{exp}\bigg\{\frac{-\lambda^2}{2}\bigg\} d\lambda.
	\vspace{-2pt}
\end{equation}

When $p_{k,i}^W$ falls below the threshold $\tau_i^W = \alpha_i^{\text{des}}$, i.e., $p_{k,i}^W < \tau_i^{W}$, we reject $\mathcal{H}_0$ from \eqref{eq:Wilcoxon_Hypothesis} and an alarm $\psi_{k,i}^W = 1$ is triggered, otherwise $\psi_{k,i}^W =0$. In the absence of attacks, the alarm rate $\alpha_i^W$ for an $i^{th}$ sensor should be approximately the same as the desired false alarm rate $\alpha_i^W \sim \alpha_i^{\text{des}}$. Computation of $\alpha_i^W$ is over the sliding window $T^{\alpha} = (k-\ell^{\alpha}+1,k)$ of length $\ell^{\alpha}$ by $\alpha_i^W = \frac{1}{\ell^{\alpha} } \sum_{j = k-\ell^{\alpha}+1}^k \psi_{j,i}^W $. Conversely, an attack that affects the residual distribution symmetry, triggering the alarm $\psi_{k,i}^W$ more frequently, causing an elevation of alarm rate $\alpha_i^W$. For alarm rates exceeding a user defined alarm rate threshold, i.e. $\alpha_i^W > \alpha_i^{\tau}$, the $i^{th}$ sensor is deemed compromised. In the following lemma we provide a proof for bounds of the WSR test variables \eqref{eq:WilcoxonTestStat} to satisfy a desired false alarm rate $\alpha_i^{\text{des}}$.
\begin{lemma}
\label{lem:W_Lemma}
Given the residual $r_{k,i}$ for an $i^{th}$ sensor over a monitoring window $T$ consisting of $\ell$ residuals and desired false alarm rate $\alpha_i^{\text{des}}$, an alarm is triggered by the WSR test when $ \Omega^W_- \leq \{W_{k,i}^{\pm}\} \leq \Omega^W_+$ is not satisfied where
\vspace{-2pt}
\begin{equation} 
\label{eq:W_lemma}
    \hspace*{-5pt} \Omega^W_{\pm} \hspace*{-2pt}= \hspace*{-1pt} \pm | \Phi^{-1} \hspace*{-1pt} ( \alpha_i^{\text{des}}\hspace*{-2pt}/2) | \hspace*{-1pt} \sqrt{ \hspace*{-1pt} (\ell^2 \hspace*{-1pt} + \hspace*{-1pt} \ell)(2\ell \hspace*{-1pt}+\hspace*{-1pt} \hspace*{-1pt}1) \hspace*{-1pt}/24} \hspace*{-1pt}+\hspace*{-1pt} (\ell^2 \hspace*{-1pt} \hspace*{-1pt}+ \hspace*{-1pt}\ell) \hspace*{-1pt}/4.
\end{equation}
\end{lemma}
\vspace{3pt}
\begin{proof}
From the Wilcoxon test statistic equaling the sum of ranks in \eqref{eq:WilcoxonTestStat}, we can rearrange \eqref{eq:WilcoxonZ} such that $\min(W_{k,i}^{\pm}) = Z_{k,i}^{W_{\text{crit}}}\sqrt{(\ell^2+\ell)(2\ell+1)/24} + (\ell^2+\ell)/4$ where $Z_{k,i}^{W_{\text{crit}}} = \Phi^{-1}(\alpha_i^{\text{des}}/2)$ is the critical value of $Z_{k,i}^W$ for $\min(W_{k,i}^{\pm})$ satisfying a desired alarm rate $\alpha_i^{\text{des}}$ to not reject \eqref{eq:Wilcoxon_Hypothesis}. The lower bound of $\{W_{k,i}^{-},W_{k,i}^{+}\}$ must satisfy
\vspace{-3pt}
\begin{equation}\label{eq:proof1_1}
\begin{split}
\Omega^W_- =& \hspace*{3pt} \Phi^{-1}(\alpha_i^{\text{des}}/2) \sqrt{(\ell^2+\ell)(2\ell+1)/24} \\
&+ (\ell^2+\ell)/4 \leq \min(W_{k,i}^{-},W_{k,i}^{+}), 
\end{split}
\vspace{-3pt}
\end{equation}
to not sound off an alarm $\psi_{k,i}^W $. Conversely, we want to show that if the lower bound $\Omega^W_- \leq \min(W_{k,i}^{\pm})$ in \eqref{eq:proof1_1} holds then the upper bound $\Omega^W_+$ holds as well. By again manipulating \eqref{eq:WilcoxonZ} such that we take the maximum $\max(W_{k,i}^{\pm}) = Z_{k,i}^{W_{\text{crit}}}\sqrt{(\ell^2+\ell)(2\ell+1)/24} + (\ell^2+\ell)/4$ where this time $Z_{k,i}^{W_{\text{crit}}} = \Phi^{-1}(1 - \alpha_i^{\text{des}}/2)$ is the critical value of $Z_{k,i}^W$ for the upper bound $\max(W_{k,i}^{\pm})$ satisfying a desired alarm rate $\alpha_i^{\text{des}}$ to not reject \eqref{eq:Wilcoxon_Hypothesis}, the upper bound is written as
\begin{equation}\label{eq:proof1_2}
\begin{split}
\Omega^W_+ = & \hspace*{3pt} \Phi^{-1}(1-\alpha_i^{\text{des}}\hspace*{-2pt}/2) \sqrt{(\ell^2+\ell)(2\ell+1)/24} \\
&+ (\ell^2+\ell)/4 \geq \max(W_{k,i}^{-},W_{k,i}^{+}), 
\end{split}
\vspace{-4pt}
\end{equation}
to not trigger the alarm $\psi_{k,i}^W$. In the calculation of the critical z-score value from the standard normal distribution $\mathcal{N}(0,1)$ to satisfy a given desired alarm rate $\alpha_i^{\text{des}}$, it is easy to show that $|\Phi^{-1}(\alpha_i^{\text{des}}\hspace*{-1pt}/2)| = \Phi^{-1}(1-\alpha_i^{\text{des}}\hspace*{-1pt}/2)$ and $\Phi^{-1}(\alpha_i^{\text{des}}\hspace*{-1pt}/2) = -|\Phi^{-1}(\alpha_i^{\text{des}}\hspace*{-1pt}/2)|$ giving the final bounds of $ \Omega^W_- \leq ( W_{k,i}^{\pm} = \{W_{k,i}^{-},W_{k,i}^{+}\} ) \leq \Omega^W_+$ as
\vspace{-1pt}
\begin{equation}\label{eq:proof1_3}
\begin{split}
\hspace*{-1pt} -& | \Phi^{-1}(\alpha_i^{\text{des}}\hspace*{-2pt}/2)| \sqrt{(\ell^2 \hspace*{-1pt}+\hspace*{-1pt}\ell)(2\ell \hspace*{-1pt}+\hspace*{-1pt}1)/24} + (\ell^2 \hspace*{-1pt}+\hspace*{-1pt}\ell)/4 \leq \\
 & W_i^{\pm} \hspace*{-1pt} \leq \hspace*{-1pt} |\Phi^{-1}(\alpha_i^{\text{des}}\hspace*{-2pt}/2)| \sqrt{ (\ell^2 \hspace*{-1pt}+ \hspace*{-1pt}\ell)(2\ell \hspace*{-1pt}+ \hspace*{-1pt}1)/24} \hspace*{-1pt}+\hspace*{-1pt} (\ell^2 \hspace*{-1pt}+ \hspace*{-1pt}\ell)/4, \nonumber
\end{split}
\vspace{-4pt}
\end{equation}
satisfying the bounds of $\Omega^W_{\pm}$ in \eqref{eq:W_lemma}. With this we conclude that if $\min(W_{k,i}^{\pm})$ does not satisfy \eqref{eq:proof1_1} then $ \Omega^W_- \leq \{W_{k,i}^{-},W_{k,i}^{+}\} \leq \Omega^W_+$ is not satisfied, triggering alarm $\psi_{k,i}^W$ for a desired false alarm rate $\alpha_i^{\text{des}} $, ending the proof.
\end{proof}

\subsection{Serial Randomness Monitor}
\label{sec:Serial Test}

The WSR test alone is not sufficient to test for randomness, since an attacker could manipulate measurements by creating specific patterns to avoid detection on the WSR test. To test further, we need to determine if the sequence of residuals are being received randomly by leveraging the Serial Independence runs (SIR) test \cite{serial_test}. The SIR test examines the number of runs that occur over the sequence, where a ``run" is defined as one or more consecutive residuals that are greater or less than the previous value. A random sequence of residuals over a given window length should exhibit a specific expected number of runs: too many or too few number of runs would not satisfy random sequential behavior. A hypothesis test is formed with $\mathcal{H}_0$ for the absence of sensor attacks and $\mathcal{H}_a$ where attacks are present by 
\vspace{-2pt}
\begin{equation}
    \label{eq:Serial_Hypothesis}
    \begin{array}{ll}
    \begin{aligned}
	\hspace{4pt} \mathcal{H}_0\text{: } N_R = \mathrm{E}[N_R], \hspace{14pt} & \mathcal{H}_a\text{: } N_R \ne \mathrm{E}[N_R],
    \end{aligned}
    \end{array}
\vspace{-2pt}
\end{equation}
where $N_R$ is the number of observed runs, to determine whether the number of runs satisfy a randomly behaving sequence. First, we take the difference of residuals at time instances $k$ and $k-1$ over a window $T'$
\vspace{-3pt}
\begin{equation}
    \label{eq:Calc_Runs}
	\bm{r}_{T',i}' := r_{k,i}'= r_{k,i} - r_{k-1,i} \text{ }, \text{ } k \in T',
	\vspace{-1pt}
\end{equation}
where $T' = \{k -\ell+2, \dots,k \} = T \setminus \{k-\ell+1\}$ is the monitor window $T$ shortened by one by removing the oldest time instance. This in turn gives us $\ell' = \ell-1$ residual differences.
\begin{remark}
\label{rem:remark2}
A residual difference $r_{k,i}'~=~0$, $k \in T'$ from \eqref{eq:Calc_Runs} is not considered in the test and the size of $\ell'$ is reduced accordingly, i.e., $\ell'=\ell'-1$.
\end{remark}

From the sequence of residual differences \eqref{eq:Calc_Runs}, we take the sign of each residual within the window $T'$,
\vspace{-2pt}
\begin{equation}
\label{eq:Calc_RunsSigns}
	\text{sign}(\bm{r}_{k,i}'), \text{ } k \in T',
\vspace{-2pt}
\end{equation}
forming a sequence of $\ell'$ positive and negative signs. The number of runs $N_R$ are observed over the sequence of $\ell'$ residual differences. The expected mean and variance of runs \cite{serial_test} are computed by
\begin{equation}
\label{eq:SerialMeanVar}
    \begin{array}{ll}
    \begin{aligned}
    \hspace{4pt} \mathrm{E}[N_R] = \frac{2\ell'-1}{3}, \hspace{14pt} & \mathrm{Var}[N_R] = \frac{16\ell'-29}{90}.
    \end{aligned}
    \end{array}
    \vspace{-1pt}
\end{equation}

Assuming large data sets (i.e. window length $\ell \geq 25$) \cite{serial_test}, the distribution of $N_R$ converges to a normal distribution as $\ell' \to \infty$ and can be approximated to a zero mean unit variance standard normal distribution $N_R \sim \mathcal{N}(0,1)$. From the number of observed runs $N_R$ and number of residual differences $\ell'$, we compute the z-score test statistic for Serial Independence from a standard normal distribution
\vspace{-2pt}
\begin{equation}
    \label{eq:Z_approx_serial}
	Z_{k,i}^S = \frac{N_R- \mathrm{E}[N_R]}{\sqrt{\mathrm{Var}[N_R]}} = \frac{N_R-\big(2\ell'-1\big)/3}{\sqrt{\big(16\ell'-29\big)/90}}.
	\vspace{-3pt}
	\end{equation}
Using the z-score from \eqref{eq:Z_approx_serial} we compute the p-value of the observed signed residual differences by
\vspace{-2pt}
\begin{equation}
\label{eq:p_serial}
	p_{k,i}^S = \Phi(|Z_{k,i}^S|) = 2\cdot\frac{1}{\sqrt{2\pi}}\int_{|Z_{k,i}^S|}^{\infty} \text{exp}\bigg\{\frac{-|\lambda|^2}{2}\bigg\} d\lambda.
	\vspace{-1pt}
\end{equation}

When $p_{k,i}^S < \tau_i^S$ is satisfied where $\tau_i^S = \alpha_i^{\text{des}}$ denotes the threshold, we reject the null hypothesis $\mathcal{H}_0$ from \eqref{eq:Serial_Hypothesis} and an alarm $\psi_{k,i}^S =1$ is triggered. In the absence of attacks, the alarm rate $\alpha_i^S$ is approximately the same as the desired false alarm rate $\alpha_i^S \sim \alpha_i^{\text{des}}$. Alarm rate $\alpha_i^S$ over the sliding window $T^{\alpha}$ is computed by $\alpha_i^S = \frac{1}{\ell^{\alpha} } \sum_{j = k-\ell^{\alpha}+1}^k \psi_{j,i}^S $. Alarm rates exceeding a user defined alarm rate threshold, i.e. $\alpha_i^S > \alpha_i^{\tau}$, signifies that the $i^{th}$ sensor is compromised.
\begin{remark}
\label{rem:serial}
    A special case of triggering alarm $\psi_{k,i}^S = 1$ is when Remark \ref{rem:remark2} is satisfied, when two consecutive residuals are equal. Since $ r_{k,i} \sim \mathcal{N}(0,\sigma^2_i)$, the probability of having two residuals of the same value is equal to $0$. 
\end{remark}

The following lemma provides a proof for bounds of $N_R$ in the SIR test to satisfy a desired false alarm rate $\alpha_i^{\text{des}}$.
\begin{lemma}
\label{lem:SI_Lemma}
Given the residual differences $r_{k,i}' = r_{k,i}- r_{k-1,i} $ for an $i^{th}$ sensor over a window $T'$ and desired false alarm rate $\alpha_i^{\text{des}}$, an alarm is triggered by the SIR test when $ \Omega^S_- \leq N_R \leq \Omega^S_+$ is not satisfied where
\vspace{-2pt}
\begin{equation} 
\label{eq:SI_lemma}
    \hspace*{-3pt} \Omega^S_{\pm} \hspace*{-1pt}=\hspace*{-1pt} \pm | \Phi^{-1} ( \alpha_i^{\text{des}} \hspace*{-2pt}/2 ) | \sqrt{(16\ell' \hspace*{-2pt} - \hspace*{-2pt} 29)/90} + (2\ell' \hspace*{-1pt}- \hspace*{-1pt}1)/3. 
    \vspace{-2pt}
\end{equation}
\end{lemma}
\vspace{4pt}

\begin{proof}
With an observed number of runs $N_R$ within a window of $\ell'$ residual differences, we can rearrange \eqref{eq:Z_approx_serial} such that $N_R= |Z_{k,i}^S|\sqrt{(16\ell'-29)/90} + (2\ell'-1)/3$ where $|Z_{k,i}^S| = |\Phi^{-1}(\alpha_i^{\text{des}}/2)|$, we find the bounds of $N_R$ to not reject \eqref{eq:Serial_Hypothesis} for a desired false alarm rate $\alpha_i^{\text{des}}$ are
\vspace{-2pt}
\begin{equation}\label{eq:proof2}
\begin{split}
\hspace*{-5pt} -&| \Phi^{-1}( \alpha_i^{\text{des}} \hspace*{-2pt} /2) | \sqrt{(16\ell' \hspace*{-1pt} - \hspace*{-1pt}29)/90} \hspace*{-1pt}+\hspace*{-1pt} (2\ell' \hspace*{-1pt}-\hspace*{-1pt}1)/3 \leq N_R \\
&\leq | \Phi^{-1}( \alpha_i^{\text{des}} \hspace*{-2pt}/2) | \sqrt{(16\ell'\hspace*{-1pt}- \hspace*{-1pt} 29)/90} \hspace*{-1pt}+\hspace*{-1pt} (2\ell' \hspace*{-1pt}- \hspace*{-1pt}1)/3. 
\end{split}
\vspace{-2pt}
\end{equation}
From \eqref{eq:proof2} we can finally obtain the bounds of $\Omega^S_{\pm}$ in \eqref{eq:SI_lemma} for alarm triggering at a desired false alarm rate $\alpha_i^{\text{des}} $.
\end{proof}

\end{section}
\begin{section}{Stealthy Attack Analysis}
\label{sec:State Degradation}

This section analyzes the advantages of including the proposed randomness monitoring framework into well known boundary/bad-data attack detectors. To this end, we first introduce two well known anomaly (boundary) detectors -- Bad-Data \cite{BadData} and Cumulative Sum \cite{CUSUM1} detectors -- and analyze the effects of stealthy attacks on a system with and without our Randomness Monitor.

\subsection{Boundary Detectors}
\label{sec:BoundaryDetectors}

To show that our framework can easily be augmented with any detector that provides magnitude boundaries, we consider two different boundary detectors found in the CPS security literature. Both boundary detectors discussed in this section leverage the absolute value of the residual \eqref{eq:Residual} for attack detection. Consequently, in the absence of attacks (i.e. $\bm{\xi}_k = \bm{0}$), this leads to $|r_{k,i}|$ following a half-normal distribution \cite{half_normal} defined by
 \vspace{-1pt}
\begin{equation}\label{eq:half_normal}
    \mathrm{E}[|r_{k,i}|] = \sqrt{2/\pi}\sigma_i, \text{ }\mathrm{ Var}[|r_{k,i}|] = \sigma^2_i(1 - 2/\pi).
    \vspace{-2pt}
\end{equation}
where $\sigma^2_i$ was defined as the $i^{th}$ diagonal element in \eqref{eq:Residual_Covariance}.

The first detector that we consider is the \textit{Bad-Data Detector} (BDD) \cite{BadData}, a benchmark attack detector to find anomalies in sensor measurements, alarming when the residual error goes beyond a threshold. Similar to our detection framework in Section \ref{sec:framework}, the BDD can also be tuned for a desired false alarm rate $\alpha_i^{\text{des}}$. Considering the residual $r_{k,i}$ in \eqref{eq:Residual}, the BDD procedure for each $i^{th}$ sensor is as follows:

\vspace{3pt}
\centerline{\textbf{Bad-Data Detector Procedure}}
\vspace{2pt}
\hrule
\vspace{-6pt}
  \begin{equation}
  \label{eq:BadData}
        \textbf{If } |r_{k,i}| > \tau_i^B, \text{ then alarm } \psi_{k,i}^B = 1, \text{ } i \in \mathcal{S},
  \end{equation}
  \vspace{-10pt}
  \hrule
\vspace{4pt}
  
Assuming the system is without attacks, the tuned threshold $\tau_i^B$ for the BDD in \eqref{eq:BadData} with $r_{k,i} \sim \mathcal{N}(0,\sigma_i^2)$ is solved by $\tau_i^B = \sqrt{2}\sigma_i \mathrm{erf}^{-1}(1-\alpha_i^{\text{des}})$ where $\mathrm{erf}^{-1}(\cdot)$ is the \textit{inverse error function}, resulting in $\alpha_i^B \sim \alpha_i^{\text{des}}$.

The second well-known boundary detector that we consider is the \textit{CUmulative SUM} (CUSUM), which has been shown to have tighter bounds on attack detection than the BDD \cite{CUSUM1}. The CUSUM leverages the absolute value of the residual in the detection procedure and is solved by

\vspace{3pt}
\centerline{\textbf{CUSUM Detector Procedure}}
\vspace{2pt}
\hrule
\vspace{-3pt}
  \begin{equation}
  \label{pro:CUSUM}
      \begin{array}{ll}
        \hspace*{-7pt} \textbf{Initialize } S_{1,i} = 0, \text{ } i \in \mathcal{S}, & \\
        \hspace*{-7pt} S_{k,i} = \max(0,S_{k-1,i}+|r_{k,i}| \hspace*{-1pt}- \hspace*{-1pt} b_i), & \hspace*{-5pt} \textbf{if } S_{k-1,i} \leq \tau^C_i, \\
        \hspace*{-7pt} S_{k,i} = 0 \text{ and Alarm } \psi_{k,i}^C = 1, & \hspace*{-5pt} \textbf{if } S_{k-1,i} > \tau^C_i.
      \end{array}
  \end{equation}
  \vspace{-4pt}
  \hrule
  \vspace{5pt}
  
The working principle of of this detector is to accumulate the residual sequence in $S_{k,i}$, triggering an alarm $\psi_{k,i}^C = 1$ when the test variable surpasses the threshold $\tau^C_i$. A detailed explanation of how to tune threshold $\tau^C_i$ given a bias $b_i$ for a desired false alarm rate $\alpha_i^{\text{des}}$ can be found in \cite{CUSUM1}.

\subsection{State Deviation under Worst-case Stealthy Attacks}
\label{sec:StealthyAttack}

We consider the reference tracking feedback controller
\vspace{-1pt}
\begin{equation}
\label{eq:controller}
    \bm{u}_k = \bm{K}\hat{\bm{x}}_k + \bm{k}_r\bm{x}_{k}^{\text{ref}},
    \vspace{-2pt}
\end{equation}
where $\bm{K} \in \R^{s \times n}$ is the state feedback control gain matrix, $\bm{k}_r~\in~\R^{m \times m}$ is a reference gain for output tracking, $\bm{x}_{k}^{\text{ref}}$ is the reference state, and $\hat{\bm{x}}_k$ is the KF state estimate from \eqref{eq:Kalman}-\eqref{eq:SteadyState_P_K}. Choosing a suitable $\bm{K}$ such that $(\bm{A} + \bm{B}\bm{K})$ is stable (i.e. $\rho[\bm{A} + \bm{B}\bm{K}] < 1$, where $\rho[\cdot]$ is the spectral radius) and $(\bm{A},\bm{C})$ is assumed stabilizable, the closed-loop system can be written within terms of the KF estimation error as
\vspace{-2pt}
\begin{equation}
\label{eq:closed_loop_system}
    \begin{array}{l}
    \hspace*{-2pt} \bm{x}_{k+1} \hspace*{-2pt} = \hspace*{-2pt} (\bm{A} \hspace*{-1pt} + \hspace*{-1pt} \bm{B}\bm{K})\bm{x}_{k} \hspace*{-1pt}  + \hspace*{-1pt} \bm{B}\bm{k}_r\bm{x}_{k}^{\text{ref}} \hspace*{-1pt} - \hspace*{-1pt} \bm{B}\bm{K}\bm{e}_k \hspace*{-1pt} + \hspace*{-1pt} \bm{\nu}_k, \\
    \hspace*{-2pt} \bm{e}_{k+1} \hspace*{-2pt} = \hspace*{-2pt} (\bm{A} \hspace*{-1pt} - \hspace*{-1pt} \bm{L}\bm{C})\bm{e}_{k} - \bm{L}(\bm{\xi}_k + \bm{\eta}_k) + \bm{\nu}_k.
    \end{array}
    \vspace{-2pt}
\end{equation}
As an attacker injects signals into the measurements (i.e. $\bm{\xi} \ne \bm{0}$), system dynamics are indirectly affected via the interconnected term $\bm{B}\bm{K}\bm{e}_k$ from the estimation error dynamics.

In the remaining of this section we describe the maximum damage that can occur due to worst-case scenario stealthy sensor attacks. We assume the attacker has perfect knowledge of system dynamics, detection procedures, and state estimation. The objective of an attacker is to cause maximum damage to the system state by injecting attack signals $\bm{\xi}_k$ to measurements while also remaining undetected. With only the BDD implemented, the effects of a worst-case scenario attack while not triggering an alarm can be derived from \eqref{eq:Residual} and \eqref{eq:BadData} with a sustained attack signal
\vspace{-3pt}
\begin{equation}
\label{eq:BD_attack1}
    \xi_{k,i} = -\bm{C}_i\bm{e}_k - \eta_{k,i} + \tau_i^B,
    \vspace{-2pt}
\end{equation}
causing the residual $|r_{k,i}| \hspace*{-1pt} = \hspace*{-1pt} \tau_i^B $ to not trigger the BDD alarm.

Now considering CUSUM as a stand-alone detector, an adversarial wants to avoid attack vectors such that the monitoring test variable exceeds threshold $\tau^C_i$, thereby causing a reset $S_{k,i} = 0, \text{ if } S_{k-1,i} > \tau^C_i$ in \eqref{pro:CUSUM} by satisfying the CUSUM procedure sequence $S_{k,i} = \max(0,S_{k-1,i}+|\bm{C}_i\bm{e}_k + \eta_{k,i} + \xi_{k,i}|-b_i) \leq \tau^C_i$ if $S_{k-1,i} \leq \tau^C_i$. For maximum effect on state deviation, the attacker saturates the CUSUM test statistic such that $S_{k,i} = \tau^C_i$ to achieve no alarm sequences. Here we define a saturation as follows:
\begin{definition}
\textit{Saturation} of a boundary detector is defined as the maximum allowable attack signal to force the residual to, but without exceeding, the detector threshold.
\end{definition}

Beginning at a time $k$, an attacker immediately saturates $S_{k,i}$ with the attack signal,
\vspace{-3pt}
\begin{equation}
\label{eq:CUSUM_attackseq1}
    \xi_{k,i} = - \bm{C}_i\bm{e}_k - \eta_{k,i} + b_i -S_{k-1,i} + \tau^C_i,
    \vspace{-3pt}
\end{equation}
followed by
\begin{equation}
\label{eq:CUSUM_attackseq2}
    \hspace*{-63pt} \xi_{k,i} = - \bm{C}_i\bm{e}_k - \eta_{k,i} + b_i.
    \vspace{-3pt}
\end{equation}
for all future time instances to hold $S_{k,i}$ at threshold $\tau^C_i$.

With the Randomness Monitor augmented with either BDD or CUSUM, an attacker can no longer hold an attack sequence to one side as described in attacks \eqref{eq:BD_attack1}-\eqref{eq:CUSUM_attackseq2}. Rather, an attacker is forced to create an attack sequence such that $r_{k,i}$ alternates residual signs if it wants to avoid triggering alarms for both the WSR and SIR tests. The most effective attack for maximum state deviation with our augmented framework is to \textit{saturate} the boundary detector as often as possible, while leaving the remaining attack signals with an opposite sign with respect to the saturating attacks to force the residual to be as close as possible to zero.

From the WSR test, given a monitoring window $\ell$, the minimum number of \textit{non-saturating} attack signals $\xi_{k,i}$ to not trigger an alarm $\psi_{k,i}^W$ is
\vspace{-4pt}
\begin{equation}
\label{eq:min_sat_attack1}
    \hspace*{-3pt} \gamma_i^{\ell} \hspace*{-1pt} = \hspace*{-1pt} \min_{\ell^j} \bigg( \sum_{rank=1}^{\ell^j} \hspace*{-2pt} rank \bigg) \bigg| \sum_{rank=1}^{\ell^j} \hspace*{-2pt} rank > \min(W_i^{\pm}),
    \vspace{-2pt}
\end{equation}
in which $\ell^j \in \mathcal{L} = (1, \dots, \ell)$ and $\mathcal{L}$ is the set of all $ranks$ as introduced in Section~\ref{sec:Wilcoxon}. From \eqref{eq:min_sat_attack1}, we can then find the maximum number of \textit{saturating} attack signals by $ \beta_i^{\ell} = \ell - \gamma_i^{\ell}$.

\begin{proposition}
The maximum allowable saturating attack signal converges to $\lim_{\ell \to \infty} \frac{\beta_i^{\ell}}{\ell} = 1 - \frac{\sqrt{2}}{2} \approx .293$ for any $\alpha_i^{\text{des}}$ as shown by the dashed black line in Fig \ref{fig:SaturationWindow}.
\end{proposition}

\vspace{-2pt}
\begin{figure}[thb!]
\centering
\includegraphics[width=0.46\textwidth]{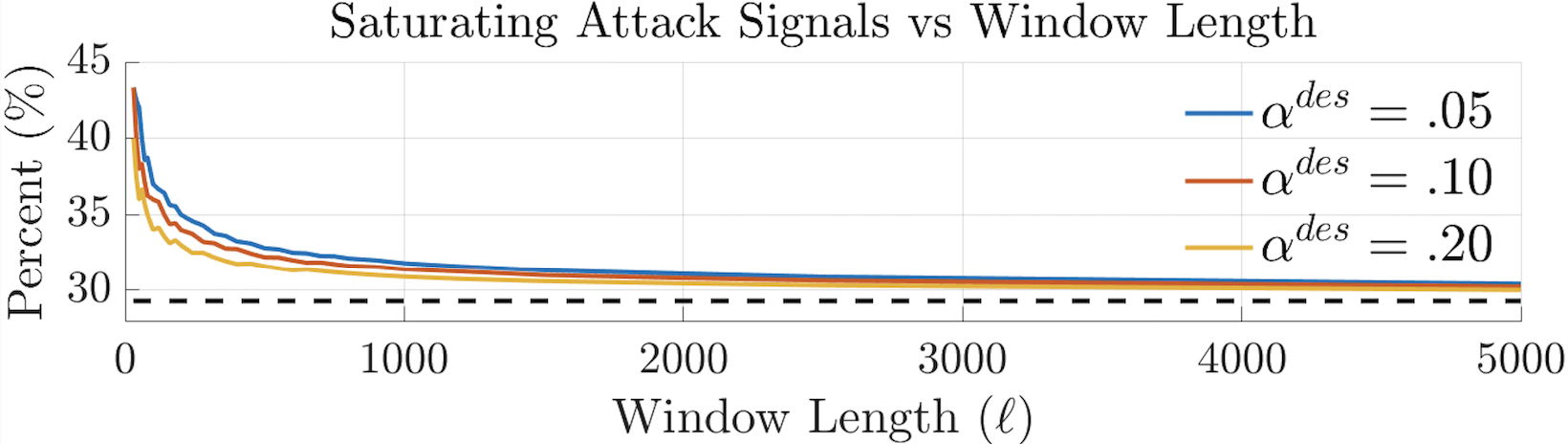}
\vspace{-5pt}
\caption{Allowable percentage of saturating attack signals of given windows lengths for different desired alarm rates $\alpha^{\text{des}}$.}
\label{fig:SaturationWindow}
\vspace{-3pt}
\end{figure}

To this point, we have discussed worst-case scenario attack sequences causing saturation of the test variable (in this paper BDD and CUSUM) to maximize the effect of the attack. However, from Remark \ref{rem:serial} in Section \ref{sec:Serial Test}, a special case to satisfy requirements of the SIR test is when two consecutive residuals of same value triggers an alarm $\psi_{k,i}^S~=~1$. To work around this issue, a stealthy attacker with perfect knowledge of the SIR test can include a small uniformly random number to the attack signal $\xi_{k,i}$ denoted by $\delta_{k,i} \sim \mathcal{U}(0,\epsilon)$ where $\epsilon \in \R^{+}$ is infinitesimally small and $\mathrm{E}[\delta_{k,i}] = \frac{\epsilon}{2} \approx 0$.  Thus, the worst-case scenario with the Randomness Monitor augmented to the BDD follows
\vspace{-1pt}
\begin{equation}
\begin{split}
    \label{eq:BD_attack_updated}
    \hspace*{-8pt} \bigg\{ \begin{array}{ll}
	\hspace*{-3pt} \xi_{k,i} \hspace*{-1pt}=\hspace*{-1pt} -\bm{C}_i\bm{e}_k - \eta_{k,i} + \tau_i^B - \delta_{k,i}, & \hspace*{-6pt} \textbf{if } \text{saturating}, \\
    \hspace*{-3pt} \xi_{k,i} \hspace*{-1pt}=\hspace*{-1pt} -\bm{C}_i\bm{e}_k - \eta_{k,i} - \delta_{k,i}, & \hspace*{-6pt} \textbf{if } \text{non-saturating},
    \end{array}
\end{split}
\vspace{-1pt}
\end{equation} 
in order to not trigger an alarm. Similarly, but with the CUSUM detector, an undetectable attack sequence follows
\vspace{-1pt}
\begin{equation}
\begin{split}
    \label{eq:CUSUM_attack_updated}
    \hspace*{-6pt} \Bigg\{ \begin{array}{ll}
	\begin{array}{l} \begin{split} \hspace*{-8pt} \xi_{k,i} = &-S_{k-1,i} - \bm{C}_i\bm{e}_k \\ &- \eta_{k,i} + b_i + \tau^C_i - \delta_{k,i},  \end{split} \end{array} & \hspace*{-5pt} \textbf{if } \text{saturating}, \\
    \hspace*{-3pt} \xi_{k,i} = - \bm{C}_i\bm{e}_k - \eta_{k,i} + b_i - \delta_{k,i}, & \hspace*{-5pt} \textbf{if } \text{non-saturating}.
    \end{array}
\end{split}
\vspace{-1pt}
\end{equation}

Given the alternating signed sequence of residuals over the monitoring window, the expected value of $r_{k,i}$ under worst-case scenario stealthy attacks is denoted as
\vspace{-1pt}
\begin{equation}
\begin{split}
    \label{eq:Expected_Residual}
    \hspace*{-7pt} \bigg\{ \begin{array}{ll}
	\hspace*{-3pt} \mathrm{E}[r^B_{k,i}] =  \tau_i^B(\frac{\beta_i^{\ell}}{\ell} - \delta_{k,i}) \approx \tau_i^B\frac{\beta_i^{\ell}}{\ell}, & \hspace*{-5pt} \textbf{for } \text{Bad-Data}, \\
    \hspace*{-3pt} \mathrm{E}[r^C_{k,i}] = \tau^C_i(\frac{\beta_i^{\ell}}{\ell} - \delta_{k,i}) \approx \tau^C_i\frac{\beta_i^{\ell}}{\ell}, & \hspace*{-5pt} \textbf{for } \text{CUSUM}.
    \end{array}
\end{split}
\vspace{-1pt}
\end{equation}

With our framework augmented to the BDD, the expected value of the residual sequence is described as $\mathrm{E}[\bm{r}^B_k] = (\mathrm{E}[r^B_{k,1}], \dots, \mathrm{E}[r^B_{k,s}])^T$ and the expectation of the closed-loop system \eqref{eq:closed_loop_system} under attack \eqref{eq:BD_attack_updated} results in
\begin{equation}
\label{eq:BD_closed_loop}
    \begin{array}{l}
    \begin{aligned}
    \mathrm{E}[\bm{x}_{k+1}] &= (\bm{A} + \bm{B}\bm{K})\mathrm{E}[\bm{x}_{k}] - \bm{B}\bm{K}\mathrm{E}[\bm{e}_k],  \\
    \mathrm{E}[\bm{e}_{k+1}] &= \bm{A}\mathrm{E}[\bm{e}_{k}] - \bm{L}\mathrm{E}[\bm{r}^B_k].
    \end{aligned}
    \end{array}
    \vspace{-3pt}
\end{equation}
Note, in \eqref{eq:BD_closed_loop}, the reference signal term $\bm{B}\bm{k}_r\bm{x}_{k}^{\text{ref}}$ from \eqref{eq:closed_loop_system} has been removed as we are interested in the expected state deviation under an attack. It is clear that if $\rho[\bm{A}]>1$ and $\mathrm{E}[\bm{r}_k^{B}] \ne \bm{0}$ then the expectation of the estimation error $\mathrm{E}[\bm{e}_{k}]$ for destabilized states diverge to infinity as $k \to \infty$ (depending on algebraic properties of $\bm{A}$), indirectly causing these states within $\mathrm{E}[\bm{x}_{k}]$ to also diverge unbounded. 
\begin{lemma}
\label{lemma3} 
    Considering a closed-loop system from \eqref{eq:system1} and \eqref{eq:BD_closed_loop}, where $\rho[\bm{A}] < 1$ and attack sequence in \eqref{eq:BD_attack_updated}, the limit for expected state divergence $\lim_{k \to \infty}\mathrm{E}[\bm{x}_k]= \Delta^{B}$ is
    \vspace{-2pt}
    \begin{equation}
    \label{eq:max_dev_BD}
        \Delta^{B} = (\bm{I}-\bm{A}-\bm{BK})^{-1}\bm{BK}(\bm{I}-\bm{A})^{-1}\bm{L}\mathrm{E}[\bm{r}^B_k].
        \vspace{2pt}
    \end{equation}
\end{lemma}
\begin{proof}
\label{proof3}
    Assuming both $\rho[\bm{A}] < 1$ and $\rho[\bm{A} + \bm{BK}] < 1$ are satisfied, signifying the invertibility of $(\bm{I}- \bm{A})$ and $(\bm{I}- \bm{A}- \bm{BK})$ in \eqref{eq:max_dev_BD}, an expected equilibrium is reached as $k \to \infty$ by $\mathrm{E}[\bm{x}_{\infty}] = (\bm{I} - \bm{A} - \bm{BK})^{-1}\bm{BK}(\bm{I} - \bm{A})^{-1} \bm{L}\mathrm{E}[\bm{r}_k^{B}]$ and $\mathrm{E}[\bm{e}_{\infty}] = (\bm{I} -\bm{A})^{-1}\bm{L}\mathrm{E}[\bm{r}_k^{B}]$ such that the evolution of \eqref{eq:BD_closed_loop} with the expected differences $\mathrm{E}[\bm{x}_k]- \mathrm{E}[\bm{x}_{\infty}]$ and $\mathrm{E}[\bm{e}_k]- \mathrm{E}[\bm{e}_{\infty}]$ is described by 
    \vspace{-3pt}
    \begin{equation}
    \label{eq:proof_max_dev_BD2}
    \begin{split}
        \hspace*{-3pt} \mathrm{E}[\bm{x}_{k+1}]-\mathrm{E}[\bm{x}_{\infty}] = & \hspace*{3pt} (\bm{A}+\bm{BK})(\mathrm{E}[\bm{x}_{k}]-\mathrm{E}[\bm{x}_{\infty}]) \\ &- \bm{BK}(\mathrm{E}[\bm{e}_{k}]-\mathrm{E}[\bm{e}_{\infty}]), \\
        \vspace{-1pt}
        \mathrm{E}[\bm{e}_{k+1}]-\mathrm{E}[\bm{e}_{\infty}] = & \hspace*{3pt} \bm{A} \mathrm{E}[\bm{e}_{k}]-\mathrm{E}[\bm{e}_{\infty}], 
    \end{split}
    \vspace{-4pt}
    \end{equation}
    are asymptotically stable i.e., $\lim_{k \to \infty}(\mathrm{E}[\bm{x}_{k+1}]-\mathrm{E}[\bm{x}_{\infty}]) = \bm{0}$ and $\lim_{k \to \infty}(\mathrm{E}[\bm{e}_{k+1}]-\mathrm{E}[\bm{e}_{\infty}]) = \bm{0}$, therefore concluding the proof.
\end{proof}

Similarly, with the Randomness Monitor augmented to CUSUM, the expected closed-loop system evolution under attack sequence \eqref{eq:CUSUM_attack_updated} is described by
\vspace{-3pt}
\begin{equation}
\label{eq:CUSUM_closed_loop}
    \begin{array}{l}
    \begin{aligned}
    \hspace*{-1pt} \mathrm{E}[\bm{x}_{k+1}] &= (\bm{A} + \bm{B}\bm{K})\mathrm{E}[\bm{x}_{k}] - \bm{B}\bm{K}\mathrm{E}[\bm{e}_k],  \\
    \hspace*{-1pt} \mathrm{E}[\bm{e}_{k+1}] &= \bm{A}\mathrm{E}[\bm{e}_{k}] - \bm{L}\mathrm{E}[\bm{r}^C_{k}].
    \end{aligned}
    \end{array}
    \vspace{-3pt}
\end{equation}
where $\mathrm{E}[\bm{r}^C_k] = (\mathrm{E}[r^C_{k,1}], \dots, \mathrm{E}[r^C_{k,s}])^T$ is the expected value of the residual sequence vector for CUSUM in \eqref{eq:Expected_Residual}.
\begin{lemma}
\label{lemma4}
    Considering a closed-loop system from \eqref{eq:system1} and \eqref{eq:CUSUM_closed_loop}, where $\rho[\bm{A}]~<~1$ and attack sequence in \eqref{eq:CUSUM_attack_updated}, the limit for expected state divergence $\lim_{k \to \infty}\mathrm{E}[\bm{x}_k]= \Delta^{C}$ is
    \vspace{-2pt}
    \begin{equation}
    \label{eq:max_dev_CUSUM}
        \Delta^{C} = (\bm{I}-\bm{A}-\bm{BK})^{-1}\bm{BK}(\bm{I}-\bm{A})^{-1}\bm{L}\mathrm{E}[\bm{r}^C_k]. 
    \end{equation}
\end{lemma}
\begin{proof}
\label{proof4}
    The proof is omitted here due to space constraints but follows the proof for Lemma \ref{lemma3}.
\end{proof}
\vspace{-1pt}

\end{section}
\begin{section}{Results} \label{sec:Results}
The proposed Randomness Monitor framework was validated in simulation and experiments and compared to state-of-the-art detection techniques introduced in Section~\ref{sec:StealthyAttack}. The case study presented in this paper is an autonomous waypoint navigation of a skid-steering differential-drive UGV with the following linearized model \cite{vehiclemodel}
\begin{equation}
\begin{split}
\label{eq:UGV_dynamics}
    \dot{v} &= \frac{1}{m}(F_l+F_r-B_rv), \\
    \dot{\omega} &= \frac{1}{I_z}\Big(\frac{w}{2}(F_l-F_r)-B_l\omega \Big), \text{ } \dot{\theta} = \omega,
\end{split}
\end{equation}
where $v$ is the velocity, $\theta$ is the vehicle's heading angle, and $\omega$ its angular velocity, forming the state vector $\bm{x} = [v,\theta,\omega]^T$. $F_l$ and $F_r$ describe the left and right input forces from the wheels, $w$ is the vehicle width, while $B_r$ and $B_l$ are resistances due to the wheels rolling and turning. The continuous-time model \eqref{eq:UGV_dynamics} is discretized with a sampling rate $t_s = 0.05$ to satisfy the system model described in \eqref{eq:system1}.

In both simulation and experiment we perform two different attack sequences: \textit{Attack~\#1} where a stealthy attack sequence concentrates the residual distribution with a non-zero mean and smaller variance whereas \textit{Attack~\#2} creates a signed pattern sequence \{+, +, +, -\} of residual differences $r'_{k,i}$. Both attacks are chosen to not increase the boundary detector alarm rate.
\vspace{-2pt}

\subsection{Simulations}
\label{sec:simulation}

Considering the UGV system model \eqref{eq:UGV_dynamics} in our case study, we show the effect of stealthy attacks on the velocity sensor on state $x_1$ with a monitoring window length $\ell = 100$. Table \ref{tab:Sim_results} gives the resulting alarm rates when our framework is augmented to boundary detectors (BDD and CUSUM) with all detectors tuned for desired false alarm rates $\alpha^{\text{des}} \in \{.05,.20\}$ and in separate simulations we show the alarm rate for \textit{No Attack}, \textit{Attack~\#1}, and \textit{Attack~\#2}. As expected, with no attacks present, all alarm rates converge approximately to the desired false alarm rate $\alpha_1^{\text{des}}$. Under \textit{Attack \#1}, alarm rates for only the WSR increase to higher values and similarly the \textit{Attack \#2} pattern
gives an increased alarm rate to only the SIR test. We should note that the window length $\ell$ results in different behaviors: short window lengths result in faster responses, while longer window lengths react slower but are able to detect more hidden attacks exhibiting non-random behavior than a monitor with a short window length. Fig.~\ref{fig:ExpectedStateDev} demonstrates attacks on the velocity sensor where our detectors are tuned for $\alpha_1^{\text{des}} = 0.10$ and compared with the CUSUM boundary detector. \textit{Attack~\#1} occurs between ($50,125$)s triggering the WSR test, \textit{Attack~\#2} between ($175,250$)s triggering the SIR test, and from $300$s a third attack satisfying bounds for both randomness tests but violating the CUSUM test is presented. Velocity is reduced as expected according to \eqref{eq:closed_loop_system} while experiencing the effects of each attack.
\vspace{-5pt}

\begin{table}[b!ht]
  \caption{Simulated Alarm Rates}
  \vspace{-8pt}
  \label{tab:Sim_results}
  \includegraphics[width=1\linewidth]{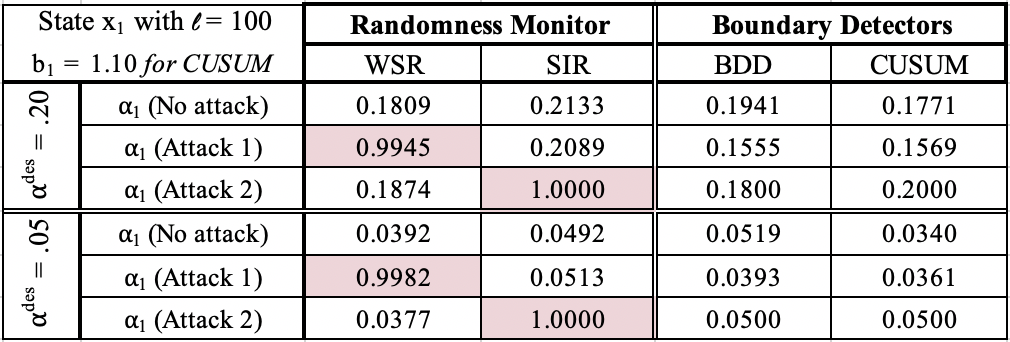}
\end{table}
\vspace{-10pt}

\begin{figure}[htb!]
\centering
\includegraphics[width=1\linewidth]{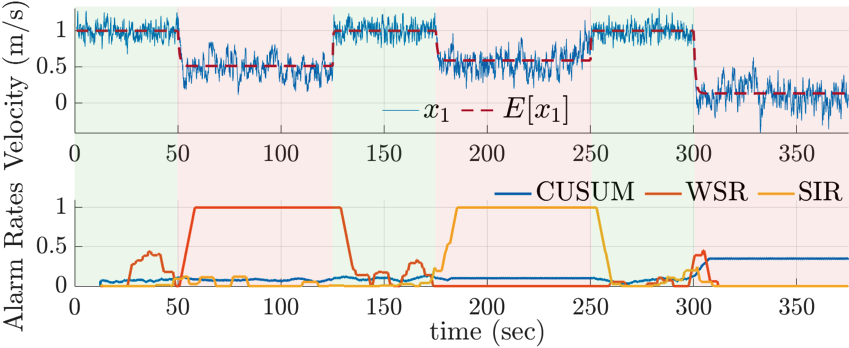}
\vspace{-17pt}
\caption{State deviation under various attacks and alarm rates over a moving window of the past $100$ time steps.}
\label{fig:ExpectedStateDev}
\vspace{-1pt}
\end{figure}

\subsection{Experiments}
\label{sec:experiment}

In this section we present a case study for a UGV performing way-point navigation under stealthy sensor attacks. For our case, the UGV travels to a series of goal positions while avoiding a restricted area with a desired cruise velocity $v^{\text{ref}}= 0.15$m/s while experiencing the same class of attacks as in Section \ref{sec:simulation}. This time the IMU sensor that measures angle $\theta$ is spoofed while our Randomness Monitor is augmented with the BDD. Fig.~\ref{fig:Experiment_ang} shows the UGV position while traveling to the four goal points. For both attacks the vehicle enters the restricted area (marked by red tape) while the boundary detector (BDD) does not see the attack in each case. The alarm rate for the WSR test increases for the case under \textit{Attack~\#1} (solid line) and the SIR test alarm rate increases during the case for \textit{Attack ~\#2} (dashed line), as expected.
\vspace{-1pt}
\begin{figure}[htb!]
\centering
\vspace{-4pt}
\includegraphics[width=1\linewidth]{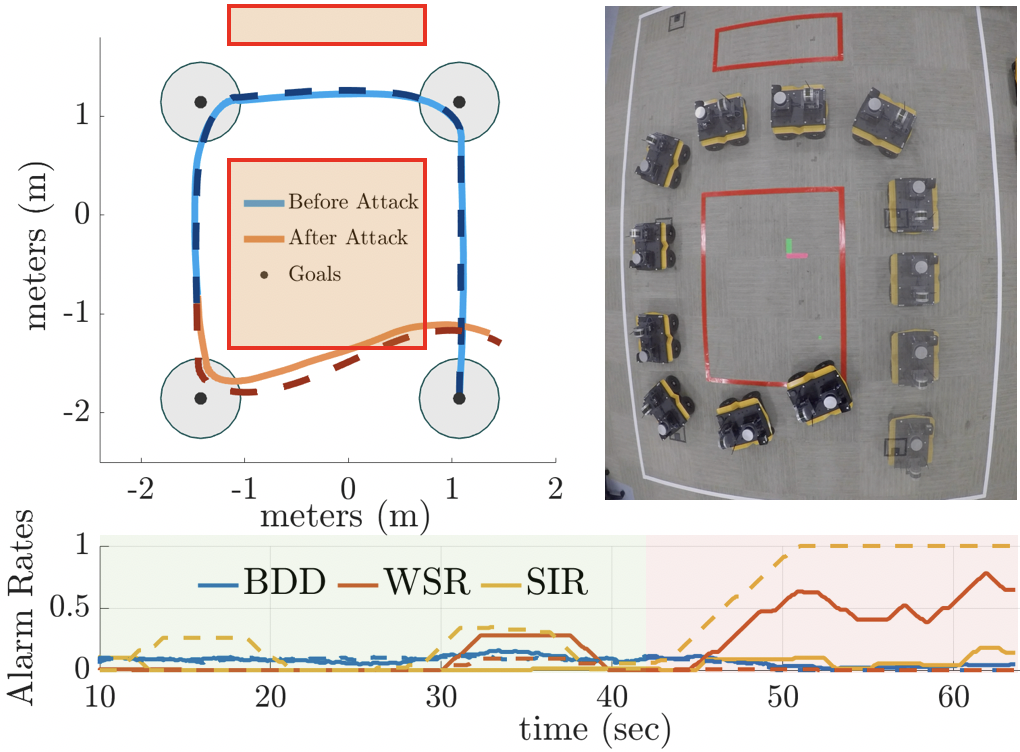}
\vspace{-17pt}
\caption{UGV position under \textit{Attack~\#1} (solid line) and \textit{Attack~\#2} (dashed line). The bottom graph displays the resulting alarm rates.}
\label{fig:Experiment_ang}
\vspace{-7pt}
\end{figure}

\end{section}
\begin{section}{Conclusions \& Future Work} 
\label{sec:conclusion}

In this paper we have proposed a monitoring framework to find cyber-attacks that present non-random behavior with the intention to hijack a system from a desired state. Our framework leverages the Wilcoxon Signed-Rank test and Serial Independence Runs test over a sliding monitor window to detect stealthy attacks when augmented to state-of-the-art boundary detectors. Among the key results of this work we provide: bounds for desired false alarm rate for each test which are leveraged to detect attacks, bounds on state deviation under worst case attack scenario, demonstrating that the proposed framework outperform detectors that solely use boundary tests. The proposed approach was validated through simulations and experiments on UGV case studies. 

In our future work we plan to extend the current work to remove this dependency from the monitoring window and plan to leverage our approach in systems with redundant sensors to remove the compromised sensors and build attack resilient controllers similar to our previous work in \cite{bezzo_SE}.
\vspace{-2pt}

\section*{Acknowledgments} 
\vspace{-2pt}
This work is based on research sponsored by ONR under agreement number N000141712012, and NSF under grant \#1816591.
\vspace{-2pt}

\end{section}

\bibliographystyle{IEEEtran}
\bibliography{ms}

\end{document}